%% file: ms.tex
\documentclass[12pt,a4paper,notitlepage,oneside]{article}

\include{header}

\usepackage{abstract_def}

\usepackage{titling}
\usepackage{authblk}
\title{On the precision matrix of an irregularly sampled AR(1) process}
\author[1]{Benjamin Allévius} 
\affil[1]{Department of Mathematics, Stockholm University, Sweden}
\date{\vspace{-5ex}}

\begin{document}
\maketitle

\begin{chapabstract} 
  \input{abstract}  
\end{chapabstract}
\thispagestyle{empty}
% \clearpage
% \tableofcontents
\thispagestyle{empty}
\clearpage
\setcounter{page}{1}
\input{content}
\bibliographystyle{apalike}
\bibliography{ms}
\end{document}

%% file: header.tex
\usepackage{alltt}
\usepackage[T1]{fontenc}
\usepackage[utf8]{inputenc}
\usepackage[english]{babel}
\usepackage{amsmath, amsthm, amssymb, mathtools}
\usepackage{dsfont}
\usepackage[cal=cm]{mathalfa}
\numberwithin{equation}{section}
\usepackage{chngcntr}
\counterwithout{equation}{section}

\usepackage{algorithm}
\usepackage[noend]{algpseudocode}

\usepackage{setspace} % Linespacing
\usepackage{graphicx}
\usepackage{adjustbox}
\usepackage[usenames,dvipsnames]{xcolor}

\usepackage{float} % Use parameter H to place float exactly where mentioned
\usepackage{wrapfig}
\usepackage{subcaption}
\usepackage[format=plain,labelfont=it,textfont=normalfont]{caption}

\usepackage{booktabs}

\usepackage{todonotes} % Use: \todo[fancyline]{text}
\usepackage{url}
\usepackage{pdflscape} % Landscape mode for certain pages
\usepackage{framed} % Use: has environments  framed, shaded, leftbar, and more
\usepackage{enumitem} % lists

\usepackage[iso,swedish]{isodate}

\usepackage[round,colon,authoryear]{natbib}

\newtheorem{theorem}{Theorem}
\newtheorem{corollary}{Corollary}[theorem]

\newcommand{\E}{\mathbb{E}} % expectation E

\newcommand{\bs}{\boldsymbol}

%% file: abstract.tex
\small{
\noindent 
% \vspace{10pt}
Irregularly sampled AR(1) processes appear in many computationally demanding applications. 
This text provides an analytical expression for the precision matrix of such a process, 
and gives efficient algorithms for density evaluation and simulation, implemented in the R package \textsf{irregulAR1}. \vspace{10pt}

}

%% file: content.tex
\section{Introduction} \label{sec:intro}
\noindent
Autoregressive (AR) processes are widely used to model time series with dependence.
One particular application is in disease surveillance, where missing data is a commonly
occurring phenomenon \citep[see e.g.][]{Gharbi2011,Sumi2011} that needs to be accounted
for in both inferential procedures and prediction.
AR processes of order one are a particularly parsimonious model choice, in which the next value of
the process depends only on the previous one. When data is missing for such a process, what
can be said about the sample of irregularly spaced values?

In this paper, we consider a zero-mean stationary autoregressive process of order one with Gaussian errors, 
which can be expressed as
\begin{align} 
  X_1 &\sim \text{Normal}\left(0, \frac{\sigma^2}{1 - \rho^2} \right), \nonumber \\
  X_t &= \rho X_{t-1} + \epsilon_t, \quad  t = 2, 3, \ldots, \label{eq:ar1def} \\
  \epsilon_t &\overset{\text{iid}}{\sim} \text{Normal}(0, \sigma^2), \nonumber
\end{align}
where $|\rho| < 1$, and, to eliminate the trivial case, $\rho \neq 0$.
We consider the zero-mean AR(1) process here because a mean term can always be added back later.
As shown in \citet[p.\ 217]{Lindsey2004} the joint distribution
of $\tilde{\bs{x}} = (X_1, X_2, \ldots, X_n)'$ is multivariate normal with zero (vector) mean
and a covariance matrix given by 
\begin{align} \label{eq:fullSigma}
 \tilde{\bs{\Sigma}} =
 \frac{\sigma^2}{1 - \rho^2}
 \begin{bmatrix}
   1          & \rho       &\dots  & \rho^{n-2} & \rho^{n-1} \\
   \rho       & 1          &\dots  & \rho^{n-3} & \rho^{n-2} \\
   \vdots     & \vdots     &\ddots & \vdots     & \vdots \\
   \rho^{n-2} & \rho^{n-3} &\dots  & 1          & \rho \\
   \rho^{n-1} & \rho^{n-2} &\dots  & \rho       & 1
 \end{bmatrix}.
\end{align}
Further, the precision matrix $\tilde{\bs{Q}} = \tilde{\bs{\Sigma}}^{-1}$ is tridiagonal and 
may be expressed as
\begin{align} \label{eq:fullQ}
 \tilde{\bs{Q}} =
 \frac{1}{\sigma^2}
 \begin{bmatrix}
   1      & -\rho    &\dots & 0         & 0 \\
   -\rho  & 1+\rho^2 &\dots & 0         & 0\\
   \vdots & \vdots   &\ddots &\vdots     &\vdots\\
   0      & 0        &\dots  & 1+\rho^2  & -\rho \\
   0      & 0        &\dots  & -\rho     & 1
 \end{bmatrix}.
\end{align}
As shown in \citet{Rue2005}, the sparsity of the precision matrix $\tilde{\bs{Q}}$ carries
over to its Cholesky decomposition, which
enables very fast (linear in $n$) density evaluation and random number
generation. In comparison, a Cholesky or eigendecomposition of the (dense) covariance matrix
has a computational complexity of order $n^3$.

In some applications it may be impossible to sample a process modeled by Equation \eqref{eq:ar1def} at 
consecutive time points. In such circumstances, it may still be of interest have a computationally convenient 
representation of the distribution of the sample $\bs{x} = (X_{t_1}, X_{t_2}, \ldots, X_{t_m})'$,
where $1 \leq t_1 < t_2 < \ldots < t_m \leq n$. 
For this irregularly sampled stationary AR(1) process, we see from the expression of 
$\tilde{\bs{\Sigma}}$ in Equation \eqref{eq:fullSigma} that the (marginal) distribution of 
$\bs{x}$ will also be multivariate normal, with the zero vector (of length $m$) as mean, and a 
covariance matrix with elements $\left(\bs{\Sigma}\right)_{ij} = \frac{\sigma^2}{1 - \rho^2} \rho^{|t_i-t_j|}$
% \begin{align} \label{eq:Sigma}
%   \bs{\Sigma} =
%  \frac{\sigma^2}{1 - \rho^2}
%  \begin{bmatrix}
%    1                  & \rho^{t_2-t_1}     &\dots  & \rho^{t_{m-1}-t_1}   & \rho^{t_m-t_1}\\
%    \rho^{t_2-t_1}     & 1                  &\dots  & \rho^{t_{m-1}-t_2}   & \rho^{t_m-t_2}\\
%    \vdots             & \vdots             &\ddots & \vdots               & \vdots\\
%    \rho^{t_{m-1}-t_1} & \rho^{t_{m-1}-t_2} &\dots  & 1                    & \rho^{t_m - t_{m-1}} \\
%    \rho^{t_m-t_1}     & \rho^{t_m-t_2}     &\dots  & \rho^{t_m - t_{m-1}} & 1 
%  \end{bmatrix}.
% \end{align}
% That is, element $(i,j)$ of this matrix is given by $\Sigma_{ij} = \frac{\sigma^2}{1 - \rho^2} \rho^{|t_i-t_j|}$. 
The aim is now to find an expression for $\bs{Q} = \bs{\Sigma}^{-1}$ similar in neatness to $\tilde{\bs{Q}}$,
with implications for density evaluation and simulation of the irregularly sampled process $\bs{x}$.

\section{Results} \label{sec:results}

\begin{theorem} \label{thm:Q}
Let $\bs{x} = (X_{t_1}, X_{t_2}, \ldots, X_{t_m})'$ be the values of the AR(1) process described
in Equation \eqref{eq:ar1def}, sampled at times $t_1 < t_2 < \ldots < t_m$.
Assume the process is in its stationary state, and for brevity of notation, assume $\sigma = 1$. 
The precision matrix $\bs{Q}$ of $\bs{x}$ then has elements
\begin{align*}
  Q_{1,1} &= \frac{1-\rho^2}{1-\rho^{2(t_2-t_1)}}, \\
  Q_{m,m} &= \frac{1-\rho^2}{1-\rho^{2(t_m-t_{m-1})}}, \\
  Q_{i,i} &= \frac{\left(1-\rho^2\right)\left(1 - \rho^{2(t_{i+1} - t_{i-1})}\right)}
                      {\left(1 - \rho^{2(t_{i} - t_{i-1})}\right) \left(1 - \rho^{2(t_{i+1} - t_{i})}\right)}, \quad 1 < i < m, \\
  Q_{i+1,i} &= Q_{i,i+1} = -\frac{\left(1 - \rho^2\right) \rho^{t_{i+1} - t_i}}
                                 {1 - \rho^{2(t_{i+1} - t_i)}}, \quad 1 \leq i < m, \\
  Q_{i+k,i} &= Q_{i,i+k} = 0, \quad k = 2,3, \ldots, m - i, \quad 1 \leq i < m - 1.
\end{align*}
$\bs{Q}$ is thus a tridiagonal matrix. 

% \begin{align} \label{eq:Q}
%   \frac{1-\rho^2}{\sigma^2}
%   &=
%   \begin{bmatrix}
%   % Row 1
%    \frac{1}{1 - \rho^{2(t_2-t_1)}}     & -\frac{\rho^{t_2-t_1}}{1-\rho^{2(t_2-t_1)}} & 0    &\dots & 0         & 0 \\
%   % Row 2
%    -\frac{\rho^{t_2-t_1}}{1-\rho^{2(t_2-t_1)}} & 
%     \frac{1 - \rho^{2(t_3-t_2)}}{(1-\rho^{2(t_2-t_1)})(1-\rho^{2(t_3-t_2)})} & 
%    -\frac{\rho^{t_3-t_2}}{1-\rho^{2(t_3-t_2)}} & 0 &0 & 0 \\
%    0  & -\frac{\rho^{t_3-t_2}}{1-\rho^{2(t_3-t_2)}} &\dots & 0         & 0\\
%    \vdots & \vdots   &\ddots &\vdots     &\vdots\\
%    0      & 0        &\dots  & 1+\rho^2  & -\rho \\
%    0      & 0        &\dots  & -\rho     & 1
%  \end{bmatrix}.
% \end{align}
\end{theorem}

\begin{proof}
  The fact that $\bs{Q}$ is a tridiagonal matrix follows from the interpretation of the off-diagonal elements of
  $\bs{Q}$ as the negated and scaled conditional correlations of $\bs{x}$ \citep[Theorem 2.2 in][]{Rue2005}:
  \begin{align}
    \text{Corr}(X_i, X_j | \bs{x}_{-ij}) &= -\frac{Q_{ij}}{\sqrt{Q_{ii} Q_{jj}}}, \quad \text{where} \label{eq:cond_corr} \\
    Q_{ii} &= \text{Prec}(X_i | \bs{x}_{-i}), \label{eq:cond_prec}
  \end{align}
  and the analogously defined $Q_{jj}$
  are the conditional precisions (variance reciprocals). Here, $\bs{x}_{-ij}$ denotes all elements of $\bs{x}$
  except the $i$th and $j$th.
  
  Now take $t_i, t_{i+1}, t_{j-1}, t_j \in \{t_1,t_2,\ldots,t_m\}$ with $i+1 \leq j-1$. 
  Note then that Equation \eqref{eq:ar1def} allows the representation
  \begin{align}
    X_{t_i} &= \rho^{-(t_{i+1} - t_i)} \left( X_{t_{i+1}} - \sum_{k=0}^{t_{i+1} - t_i - 1} \rho^k \epsilon_{t_{i+1} - k} \right) \\
    X_{t_j} &= \rho^{t_j - t_{j-1}} X_{t_{j-1}} + \sum_{k=0}^{t_j - t_{j-1} - 1} \rho^k \epsilon_{t_j - k}.
  \end{align}
  The quantities $X_{t_{i+1}}$ and $X_{t_{j-1}}$ are assumed known in the conditional correlation $\text{Corr}(X_{t_i}, X_{t_j} | \bs{x}_{-ij})$,
  and no same error term appears in both sums above. Thus, the correlation equals zero, which implies that $Q_{ij} = 0$ when
  we have at least one observation between $t_i$ and $t_j$. $\bs{Q}$ is then (at most) tridiagonal.
  
  With this knowledge, one way to determine the explicit form of the elements of $\bs{Q}$ is to solve the system of
  equations $\bs{\Sigma} \bs{Q} = \bs{I}_m$, where $\bs{I}_m$ is the $m \times m$ identity matrix. 
  For notational convenience, we define $\rho_{ji} = \rho^{j-i}$ with $j>i$.
  Because $\bs{Q}$ is symmetric ($Q_{k,k+1} = Q_{k+1,k}$) and tridiagonal, we only have $2m - 1$ unknowns to solve for, and 
  the system $\bs{\Sigma} \bs{Q} = \bs{I}_m$ can be reduced to the set of equations
  \begin{align}
  &\begin{cases} \label{eq:sys1}
    Q_{11} + \rho_{21} Q_{12} &= 1 - \rho^2 \\
    \rho_{21} Q_{11} + Q_{12} &= 0 \\
  \end{cases}, \\
  &\begin{cases} \label{eq:sys2}
    \rho_{k,k-1} Q_{k-1,k} + Q_{kk} + \rho_{k+1,k} Q_{k,k+1}   &= 1 - \rho^2\\
    \rho_{k+1,k-1} Q_{k-1,k} + \rho_{k+1,k} Q_{kk} + Q_{k,k+1} &= 0 \\
  \end{cases}, \quad 1 < k < m, \\
  &\begin{cases} \label{eq:sys3}
    \rho_{m,m-1} Q_{m-1,m} + Q_{mm} &= 1 - \rho^2. \\
  \end{cases}
\end{align}
The system \eqref{eq:sys1} may be solved for $Q_{11}$ and $Q_{12}$, after which system \eqref{eq:sys2}
can be solved iteratively for $k = 2, 3, \ldots, m - 1$ by inserting $Q_{k-1,k}$ found in the previous iteration.
Lastly, $Q_{mm}$ is easily solved for in Equation \eqref{eq:sys3} when $Q_{m-1,m}$ has been found in the previous
step.

If the assumption $\sigma = 1$ is relaxed, the non-zero elements of $\bs{Q}$ should be divided by $\sigma^{2}$ to yield the 
correct precision matrix.
% Details are given in Appendix \ref{seq:proof_details}.
\end{proof}

\noindent
Since $\bs{Q}$ has only $2m - 1$ non-zero elements, it can be constructed using only $\mathcal{O}(m)$ flops
(floating point operations). Likewise, it needs only $\mathcal{O}(m)$ space for storage if stored in 
a sparse format.

\begin{corollary} \label{corol:mean}
  Let $Y_{t_i} = X_{t_i} + \mu_{t_i}$, with $X_{t_i}$ an element of $\bs{x}$ as in Theorem \ref{thm:Q} and the $\mu_{t_i}\!\!$'s fixed.
  Also let $\bs{y} = (Y_{t_1}, \ldots, Y_{t_m})'$.
  Then by Theorem 2.3 of \citet{Rue2005}, the following conditional expected values hold for the elements of $\bs{y}$:
  \begin{align}
    \E[Y_{t_1} | \bs{y}_{-1}]\hspace{6pt} &= \mu_{t_1} + \rho^{t_2 - t_1} (y_{t_2} - \mu_{t_2}), \\
    \begin{split}
    \E[Y_{t_i} | \bs{y}_{-i}]\hspace{8pt} &= \mu_{t_i} 
                              \hspace{2pt}+ \rho^{t_i - t_{i-1}} \frac{1 - \rho^{2(t_{i+1} - t_i)}}{1 - \rho^{2(t_{i+1} - t_{i-1})}} (y_{t_{i-1}} - \mu_{t_{i-1}})  \\
                              &\hspace{33pt}+ \rho^{t_{i+1} - t_{i}} \frac{1 - \rho^{2(t_i - t_{i-1})}}{1 - \rho^{2(t_{i+1} - t_{i-1})}} (y_{t_{i+1}} - \mu_{t_{i+1}}),  
                              ~1 < i < m,
    \end{split} \\
    \E[Y_{t_m} | \bs{y}_{-m}] &= \mu_{t_m} + \rho^{t_m - t_{m-1}} (y_{t_{m-1}} - \mu_{t_{m-1}}),
  \end{align}
\end{corollary}
\noindent
where $\bs{y}_{-j}$ denotes all elements of $\bs{y}$ except the $j$th, for $1 \leq j \leq m$.
Together with the expression for the conditional precision given in Equation \eqref{eq:cond_prec}, 
Corollary \ref{corol:mean} specifies the full conditional (normal) distributions of the irregularly 
sampled AR(1) process.
% the conditional expectations for $Y_t$ in Corollary \ref{corol:mean} can be used to make statements about the 
% distribution of the process $\{Y_t\}$ at time points other than those at which observations are 

\section{Implications}  \label{sec:implications}
\noindent
Typical density evaluation and random number generation of multivariate normal variables involves
a Cholesky or eigendecomposition of the covariance matrix \citep{Barr1972}. 
For an $m \times m$ matrix, the computational cost (in terms of the number of flops) associated with 
either method is $\mathcal{O}(m^3)$ \citep{NLA}. For large $m$, this cost becomes prohibitive, and even with $m$
smaller such decompositions can become the bottleneck when they need to be performed repeatedly. 
For example, a MCMC algorithm trying to infer the distribution of $\rho$ may need
to evaluate the density of $\bs{x}$ thousands of times, if not more. However, if the sparse structure of $\bs{Q}$ can
be used, this cost can be drastically reduced.

\subsection{Cholesky factorization}  \label{sec:chol}
\noindent
Indeed, the sparsity of $\bs{Q}$ carries over to its Cholesky decomposition.
By Theorem 2.9 in \citet{Rue2005}, the (lower) Cholesky decomposition $\bs{L}$ of $\bs{Q}$,
i.e.\ $\bs{Q} = \bs{L} \bs{L}^T$, will have a lower bandwidth of 1---that is, only the main 
diagonal and (first) subdiagonal will have non-zero elements. This decomposition is computable
in linear time using Algorithm 2.9 in \citet{Rue2005}, described next. 

Let $v$ be a vector of length $m$ and let $v_{i:j}$ be elements $i$ to $j$ of this vector.
Denote by $Q_{i:k,j}$ elements $i$ to $k$ in column $j$ of $\bs{Q}$, and let the same notation be
applicable to $\bs{L}$ and its elements $L_{i,j}$, which are initialized to zero.
The matrix $\bs{L}$ can then be computed using the following algorithm:
\begin{algorithm}[H]
\caption{Band-Cholesky factorization of $\bs{Q}$ (bandwidth 1)}\label{alg:L}
\begin{algorithmic}[1]
\For{$j = 1$ to $m$}
  \State $\lambda \hspace{10.5pt} \gets \min\{j+1,m\}$
  \State $v_{j:\lambda} \gets Q_{j:\lambda,j}$
  \If{$j > 1$}
    \State $v_{j} \gets v_{j} - L_{j,j-1}^2$
  \EndIf
  \State $L_{j:\lambda,j} \gets v_{j:\lambda} / \sqrt{v_j}$
\EndFor
\State \textbf{Return} $\bs{L}$
\end{algorithmic}
\end{algorithm}

\noindent
This algorithm is seen to involve only $\mathcal{O}(m)$ flops. Additionally, if $\bs{L}$ is stored in a sparse format, 
the storage is of size $\mathcal{O}(m)$ as well.

\subsection{Unconditional simulation} \label{sec:sim}
\noindent
If we wish to sample $\bs{x} \sim \text{MVN}\left(\bs{\mu}, \bs{Q}^{-1}\right)$ as in Theorem \ref{thm:Q}, but now with a
mean vector $\bs{\mu}$, we can use the following algorithm \citep[Algorithm 2.4 in][]{Rue2005}:
\begin{algorithm}[H]
\caption{Sampling $\bs{x} \sim \text{Normal}(\bs{\mu}, \bs{Q}^{-1})$}\label{alg:uncond_sampling}
\begin{algorithmic}[1]
\State Compute $\bs{L}$ using Algorithm \ref{alg:L}.
\State Sample $m$ standard normal variables and store them in a vector $\bs{z}$.
\State Solve $\bs{L}^T \bs{v} = \bs{z}$ using sparse back substitution (see Algorithm \ref{alg:sbs} below). 
\State Compute $\bs{x} = \bs{\mu} + \bs{v}$.
\State \textbf{Return} $\bs{x}$
\end{algorithmic}
\end{algorithm}

\noindent
This algorithm is seen to be of order $\mathcal{O}(m)$ in computational complexity.
The sparse back substitution in step 3 of Algorithm \ref{alg:uncond_sampling} computes the 
elements of $\bs{v}$ as follows:
\begin{algorithm}[H]
\caption{Solving $\bs{L}^T \bs{v} = \bs{z}$ when $\bs{L}$ has bandwidth 1}\label{alg:sbs}
\begin{algorithmic}[1]
\State $v_m = z_m / L_{m,m}$
\For{$i = m-1$ to $1$}
  \State $v_i = (z_i - L_{i+1,i} v_{i+1}) / L_{i,i}$
\EndFor
\State \textbf{Return} $\bs{v}$
\end{algorithmic}
\end{algorithm}

\noindent
Only $3m-2$ flops are used to produce the solution $\bs{v}$.

\subsection{Conditional simulation} \label{sec:condsim}
\noindent
Assume now that $\bs{t}_{\text{o}} = \{t_1, \ldots, t_m\}$ are the time points at which observations $\bs{x}_{\text{o}}$ are available,
and let $\bs{t}_{\text{p}} = \{s_1, \ldots, s_k\}$ be another set of time points disjoint from $\bs{t}_{\text{o}}$.
Suppose that we wish to simulate values $\bs{x}_{\text{p}}$ from the distribution of the process at times $\bs{t}_{\text{p}}$,
conditional on $\bs{x}_{\text{o}}$. From standard facts about the multivariate normal distribution \citep[see e.g.][]{Rue2005}, we know that
\begin{align*}
  \bs{x}_{\text{p}} | \bs{x}_{\text{o}} &\sim \text{Normal}\left(\bs{\mu}_{\text{p$|$o}}, \bs{\Sigma}_{\text{p$|$o}}\right), \quad \text{where} \\
  \bs{\mu}_{\text{p$|$o}} &= \bs{\mu}_{\text{p}} + \bs{\Sigma}_{\text{po}} \bs{\Sigma}_{\text{oo}}^{-1} \left(\bs{x}_{\text{o}} - \bs{\mu}_{\text{o}} \right), \\
  \bs{\Sigma}_{\text{p$|$o}} &= \bs{\Sigma}_{\text{pp}} - \bs{\Sigma}_{\text{po}} \bs{\Sigma}_{\text{oo}}^{-1} \bs{\Sigma}_{\text{op}},
\end{align*}
and
\begin{align*}
  \bs{\mu}_{\text{a}} = 
    \begin{bmatrix}
      \bs{x}^*_{\text{p}} \\
      \bs{x}^*_{\text{o}} \\
    \end{bmatrix}
  \quad \text{and} \quad 
  \bs{\Sigma}_{\text{a}} &=
  \begin{bmatrix}
    \bs{\Sigma}_{\text{pp}} & \bs{\Sigma}_{\text{po}} \\
    \bs{\Sigma}_{\text{op}} & \bs{\Sigma}_{\text{oo}}
  \end{bmatrix}
\end{align*}
are the mean vector and covariance matrix of $\bs{x}_{\text{a}} = (\bs{x}^T_{\text{p}}, \bs{x}^T_{\text{o}})^T$.
In general, even if the matrix $\bs{\Sigma}_{\text{oo}}^{-1}$ is sparse, the computation of $\bs{\Sigma}_{\text{p$|$o}}$
will be demanding for $k \gg m$, and further so if a Cholesky or eigendecomposition of the result is to be computed as well.

A more efficient way of sampling from the distribution of $\bs{x}_{\text{p}} | \bs{x}_{\text{o}}$, especially when $k \gg m$, 
was described by \citet{Hoffman1991}.
To use this method, we need to be able to sample (unconditionally) from the joint distribution of 
$\bs{x}_{\text{a}}$. One way of doing so is to just
iteratively simulate from the definition of the process (Equation \ref{eq:ar1def}) using a starting value
drawn from the stationary distribution (and with mean terms added back to the right hand side),
and then pick out the values for the times $\bs{t}_{\text{p}} \cup \bs{t}_{\text{o}}$.
Another is to first sort $\bs{t}_{\text{p}} \cup \bs{t}_{\text{o}}$ (if necessary), order and combine the mean vectors accordingly,
create the corresponding precision matrix $\bs{Q}_{\text{a}} = \bs{\Sigma}^{-1}_{\text{a}}$, draw samples using Algorithm \ref{alg:uncond_sampling},
and re-order the samples according to $\bs{t}_{\text{p}}$ and $\bs{t}_{\text{o}}$.

Let $\bs{Q}_{\text{o}}$ be the precision matrix of $\bs{x}_{\text{o}}$, i.e.\ $\bs{Q}_{\text{o}} = \bs{\Sigma}^{-1}_{\text{oo}}$. Then we can sample
from the distribution of 
$\bs{x}_{\text{p}} | \bs{x}_{\text{o}}$ as follows \citep{Hoffman1991}:
\begin{algorithm}[H]
\caption{Sampling $\bs{x}_{\text{p}} | \bs{x}_{\text{o}} \sim \text{Normal}\left(\bs{\mu}_{\text{p$|$o}}, \bs{\Sigma}_{\text{p$|$o}}\right)$} \label{alg:cond_sampling}
\begin{algorithmic}[1]
\State Sample 
       $
         \bs{x}^*_{\text{a}} = \begin{bmatrix}
                                \bs{x}^*_{\text{p}} \\
                                \bs{x}^*_{\text{o}} \\
                              \end{bmatrix}
                            \sim
                            \text{Normal}\left( \bs{\mu}_{\text{a}}, \bs{\Sigma}_{\text{a}} \right)         
       $ as discussed above.
\State \textbf{Return} $\bs{x}_{\text{p}} = \bs{x}^*_{\text{p}} + \bs{\Sigma}_{\text{po}} \bs{Q}_{\text{o}} (\bs{x}_{\text{o}} - \bs{x}^*_{\text{o}})$.
\end{algorithmic}
\end{algorithm}
\noindent
Because $\bs{Q}_{\text{o}}$ is tridiagonal, the matrix product $\bs{\Sigma}_{\text{po}} \bs{Q}_{\text{o}}$ involves only
$\mathcal{O}(km)$ flops, rather than the $\mathcal{O}(km^2)$ flops that are needed if $\bs{Q}$ is dense.
Thus the complexity of Algorithm \ref{alg:cond_sampling} is $\mathcal{O}(km)$, plus $\mathcal{O}\left((k+m)\log(k+m)\right)$ if sorting
is to be done in Step 1.
%For $k \gg m$ and $m$ held fixed, the complexity becomes $\mathcal{O}(k\log k )$.

\subsection{Density evaluation} \label{sec:density}
\noindent
Many applications require the evaluation of probability density functions. For example, typical MCMC algorithms
calculate a quotient (or log difference) of densities repeatedly in the evaluation of acceptance ratios. 
It is therefore of interest to make this computation as efficient as possible for the type of 
irregularly sampled AR(1) process considered in this paper. Typically, the evaluation of a $m$-dimensional multivariate normal
density involves the Cholesky decomposition of the covariance matrix, which as discussed previously has a 
computational complexity of $\mathcal{O}(m^3)$. With a sparse precision matrix $\bs{Q}$ however, this 
cost can be drastically reduced. 

Let $\bs{x} \sim \text{Normal}\left(\bs{\mu}, \bs{Q}^{-1}\right)$ be a vector of values from
an irregularly sampled AR(1) process as in Theorem \ref{thm:Q}, and let $p(\bs{x})$ be probability density function evaluated
at $\bs{x}$.
To calculate $\log p(\bs{x})$ of a sample $\bs{x}$ from this distribution, first calculate the Cholesky decomposition 
$\bs{L}$ of $\bs{Q}$ using Algorithm \ref{alg:L}.
Then, by \citet[p.\ 35]{Rue2005}, the log-density can be computed as follows:
\begin{align}
  \log p(\bs{x}) &= -\frac{m}{2} \log (2 \pi) + \sum_{i=1}^m \log L_{i,i} - \frac{1}{2} q, \quad \text{where} \\
  q &= \left( \bs{x} - \bs{\mu} \right)^T \bs{Q} \left( \bs{x} - \bs{\mu} \right).
\end{align}
If $\bs{x}$ was generated using Algorithm \ref{alg:uncond_sampling}, $q$ simplifies to $q = \bs{z}^T \bs{z}$.
By utilizing the sparsity of $\bs{Q}$, the evaluation of the log-density has a computational complexity of
$\mathcal{O}(m)$.

\section{Summary}
\noindent
This paper provides analytical expressions for the elements of the precision matrix $\bs{Q}$ of a stationary Gaussian AR(1) process
sampled with irregular spacing. The sparsity of this matrix was shown in Section \ref{sec:implications} to yield efficient algorithms 
for density evaluation and simulation of such a process. Applications of AR(1) processes are abound in biostatistics and finance,
and the results of this paper should prove relevant for those in need of computational efficiency. 
More generally, the results are valuable from a missing data perspective. A simple extension of this paper is calculate exactly
how the results derived carry over to the distribution of a irregularly spaced sample from an Ornstein-Uhlenbeck process, which is the
continuous-time analog of the AR(1) process. 
A memory-efficient implementation of the given algorithms, enabled by the \textsf{RcppArmadillo} library \citep{RcppArmadillo},
is provided by the R package \textsf{irregulAR1}, 
available on the Comprehensive R Archive Network (CRAN).
% available on the \href{https://cran.r-project.org/}{Comprehensive R Archive Network} (CRAN).

\section*{Acknowledgements}
\noindent
Funding: This work was supported by the Swedish Research Council [grant 2013:05204].